\pdfoutput=1
\RequirePackage{ifpdf}
\ifpdf % We~are running pdfTeX in pdf mode
\documentclass[pdftex]{sigma}%,draft
\else
\documentclass{sigma}
\fi

\usepackage{stmaryrd}
\usepackage{tikz}
\usetikzlibrary{arrows}
\usetikzlibrary{math}

\DeclareMathOperator{\psan}{span}

\newcommand{\hilH}{\mathcal H}
\newcommand{\norm}[1]{\lVert #1\rVert}
\newcommand{\aver}[1]{\langle #1 \rangle}
\newcommand{\nN}{n\in \mathbb{N}}

\newcommand{\lN}{l\in \mathbb{N}}
\newcommand{\kN}{k\in \mathbb{N}}
\newcommand{\jN}{j\in \mathbb{N}}
\newcommand{\iN}{i\in \mathbb{N}}
\newcommand{\intt}[1]{\llbracket #1 \rrbracket}
\newcommand{\perpF}{{\perp _F}}
\newcommand{\expec}[1]{\mathbb{E}\left(#1\right)}
\newcommand{\algE}{{\mathcal E}}
\newcommand{\salg}{$\sigma$-algebra}
\newcommand{\rv}{random variable}
\newcommand{\Id}{\mathrm{Id}}
\newcommand{\ind}[1]{\mathbf{1}_{#1}}

\newtheorem{Theorem}{Theorem}[section]
\newtheorem{Corollary}[Theorem]{Corollary}
\newtheorem{Lemma}[Theorem]{Lemma}

 { \theoremstyle{definition}

 }

\begin{document}
\allowdisplaybreaks

\newcommand{\arXivNumber}{2110.07456}

\renewcommand{\PaperNumber}{012}

\FirstPageHeading

\ShortArticleName{A Quantum $0-\infty$ Law}

\ArticleName{A Quantum $\boldsymbol{0-\infty}$ Law}

\Author{Michel BAUER~$^{\rm ab}$}

\AuthorNameForHeading{M.~Bauer}

\Address{$^{\rm a)}$~Universit\'e Paris-Saclay, CNRS, CEA, Institut de Physique Th\'eorique,\\
\hphantom{$^{\rm a)}$}~91191 Gif-sur-Yvette, France}
\Address{$^{\rm b)}$~PSL Research University, CNRS, \'Ecole normale sup\'erieure,\\
\hphantom{$^{\rm b)}$}~D\'epartement de math\'ematiques et applications, 75005 Paris, France}
\EmailD{\href{mailto:michel.bauer@ipht.fr}{michel.bauer@ipht.fr}}

\ArticleDates{Received October 27, 2021, in final form February 01, 2022; Published online February 10, 2022}

\Abstract{We give conditions under which a sequence of randomly chosen orthogonal sub\-spaces of a separable Hilbert space generates the whole space.}


\Classification{60J05; 81-10}

\begin{flushright}
\begin{minipage}{47mm}
\it Dedicated to Denis Bernard\\
 on his 60\,$^{th}$ birthday
\end{minipage}
\end{flushright}

\section{Introduction}

While studying a quantum system, one is often faced with the task of comparing two bases of a separable Hilbert space. One obvious example is given by the eigenstate bases for a free and a perturbed Hamiltonian. Another example is when one basis is well-suited to describe local operators (in this basis they are diagonal, or their matrix elements decay rapidly away from the diagonal for instance) and the other is the eigenstate basis of the Hamiltonian. This is related to the eigenstate thermalization hypothesis \cite{PhysRevA.43.2046,PhysRevE.50.888} which can be -- at least partly -- rephrased as a physically motivated ansatz for matrix elements of local operators (say, in position space for instance) in the Hamiltonian eigenstate basis. The question also has some more practical aspects when some spaces of trial wave functions are used to get approximations of the low-lying energy eigenstates of a quantum system. In this setting, the dimension of the trial space is larger than the number of eigenstates which are to be approximated. In one version, all the approximate eigenstates are looked for in a single large space, but it may also be desirable to look for the first eigenstate in some trial space, for the second in a larger one and so on, leading to a nested structure like ``Russian dolls''.

Such a Matryoshka picture also paves the way to the construction of random Hamiltonians, and this is the vantage point we adopt in the following. Suppose we are given a (finite-dimensional) subspace $F_1$ of the full Hilbert space and we are to choose a ray, i.e., a one-dimensional subspace, $E_1$ inside it. There is a single unitary-invariant probability measure to make this choice at random. If one looks for another ray, orthogonal to $E_1$, or, what amounts to the same, for a two-dimensional subspace $E_2 \supset E_1$, in a larger (finite-dimensional) subspace $F_2 \supset F_1$ of the full Hilbert space, there is again a single unitary-invariant probability measure to make this choice at random. Given a sequence $\{0\}\subset F_1\subset F_2\subset\cdots$ of trial subspaces, i.e., of finite-dimensional subspaces of the full Hilbert space with $\dim F_k \geq k$ for $\kN$ such that $\cup_{\kN} F_k$ is dense in the full Hilbert space, this procedure can be continued to produce a random orthonormal sequence of rays in the Hilbert space, or what amounts to the same a sequence $\{0\}\subset E_1\subset E_2\subset\cdots$ with $E_k\subset F_k$ and $\dim E_k = k$ for $\kN$. Statistical properties of matrix elements with respect to this sequence, as compared to those in a fixed basis of the Hilbert space, can be studied. This is a difficult task. But there is a preliminary question, in fact the only one we address in this study: does this procedure really construct a basis, or do ``holes'' remain in the Hilbert space?

Section \ref{sec:defmainres} makes our framework precise. We introduce the notion of Matryoshka tailored for our purposes, and state our main result: a necessary and sufficient condition to get a complete orthonormal system, based on the growth of the dimensions of the trial spaces. Section \ref{sec:jjjparadox} reinterprets the problem as a quantum version of a classical puzzle/paradox. Section \ref{sec:mainproofs} is devoted to the proofs. Section \ref{sec:concl} mention further possible directions.

\section{Definitions and main result} \label{sec:defmainres}

Let $\hilH$ denote a complex separable infinite-dimensional Hilbert space. The norm is denoted by~$\norm{\;}$ and the inner product by $\aver{\; , \;}$, which we take to be linear in the second argument, as in quantum mechanics. If $S \subset \hilH$, we set $S^\perp:=\{v\in \hilH, \, \aver{s , v}=0 \text{ for every } s\in S\}$ as usual.

Let $\underline{n}:=(n_k)_{\kN}$ be an increasing sequence $0=n_0 \leq n_1\leq n_2\leq\cdots $ of integers with \mbox{$\lim_{k\to \infty}n_k=\infty$}.
An $\underline{n}$-Matryoshka $\underline{F}:=(F_k)_{\kN}$ is an increasing sequence $\{0\}=F_0\subset F_1\subset F_2\subset\cdots$ of subspaces of $\hilH$ with $\dim F_k:=n_k$ for $\kN$.
An $\underline{n}$-Matryoshka $\underline{F}$ is called total (in~$\hilH$) if $\cup_{\kN} F_k$ is dense in~$\hilH$.

In the following, until Section~\ref{sec:concl}, we fix once and for all the sequence $\underline{n}$ and a total $\underline{n}$-Mat\-ryoshka $\underline{F}$. Note that by the choice of an appropriate orthonormal basis $(v_n)_{n\geq 1}$ in $\hilH$ we could reduce the situation to the case when $F_k:=\psan (v_n)_{n\in \intt{1,n_k}}$. We shall sometimes exploit this possibility.
In the framework of numerical computation of low-lying eigenstates, $F_k$ would be the $k^{\text{th}}$ trial space, the space in which the first $k$ levels are to be found.

We also fix an increasing sequence $\underline{m}:=(m_k)_{\kN}$ of integers such that $m_k\leq n_k$ for $\kN$ but $\lim_{k\to \infty}m_k=\infty$. We set $(\delta m)_k:=m_{k+1}-m_k$.
An $\underline{m}$-Matryoshka in $\underline{F}$ is an $\underline{m}$-Matryoshka $\underline{E}:=(E_k)_{\kN}$ such that $E_k \subset F_k$ for $\kN$. In~the sequel, if $\underline{E}$ is an $\underline{m}$-Matryoshka in~$\underline{F}$, $E_k^\perpF$ denotes the orthogonal complement of~$E_k$ in~$F_k$ (not in~$\hilH$).
The case $m_k=k$ fits to our original motivation of finding $k$ low-lying eigenstates in $F_k$. More general sequences $(m_k)_{\kN}$ would allow to deal with degenerate energy levels for instance. Though this complicates notations a little bit, the mathematics is essentially the same.

We denote by $M_{\underline{F},\underline{m}}$, or simply $M$ when there is no risk of confusion, the set of $\underline{m}$-Matryoshkas in $\underline{F}$.
We recall that if $F$ is a finite-dimensional Hilbert space\footnote{In this study, a finite-dimensional subspace of a separable infinite-dimensional reference Hilbert space.} of dimension $n$ and $m\leq n$ there is an unique unitary invariant probability measure $\mu_F^m$ on the set of $m$-dimensional subspaces of~$F$ (a Grassmannian which we denote by $\mathbb{G}(F,m)$).

An important consequence of the definition of $\mu_F^m$ for our discussion is the following: If $v$ is a vector in~$F$, $E$~is chosen with $\mu_F^m$ and $U$ is the (orthogonal) projection of $v$ on $E$ (a random vector) then $\expec{\norm{U}^2}=\norm{v}^2\frac{m}{n}$ (taken to be $0$ if $n=0$). Due to the importance of this result, we give a quick derivation in Appendix~\ref{appsec:messubsp}. The vector $W:=v-U$ is the projection of $v$ on the orthogonal complement $E^\perp$ of $E$ in $F$, and it follows from uniqueness that the image measure of $\mu_F^m$ by the map from the Grassmannian of $m$-dimensional subspaces of $F$ to that of $n-m$-dimensional subspaces, $E\mapsto E^\perp$ or equivalently $U \mapsto W$, is $\mu_F^{n-m}$ so that $\expec{\norm{W}^2}=\norm{v}^2\frac{n-m}{n}$, which must be because $\norm{v}^2=\norm{U}^2+\norm{W}^2$ identically.

Note in passing that the large $n$ behavior of the formul\ae\ for $\expec{\norm{U}^2}$ and $\expec{\norm{W}^2}$ given above gives the intuitive reason why, as is well-known, no unitary invariant measure $\mu_F^m$ can exist when $F$ is infinite-dimensional.

We use the measures $\mu_F^m$ to endow the set $M$ of $\underline{m}$-Matryoshkas in $\underline{F}$ with a probability measure $P_{\underline{F},\underline{m}}$, or simply $P$, defined recursively as follows:
\begin{itemize}\itemsep=0pt
\item $E_0=F_0=\{0\}$.
\item $E_0,\dots,E_k$ being chosen, $E_{k+1}$ is chosen as a random (in a unitary invariant way) subspace of $F_{k+1}$ containing $E_k$. The condition of unitary invariance can be rephrased as follows: $E_{k+1}$ is the orthogonal direct sum of $E_k$ with an $(\delta m)_k$-dimensional subspace of $E_k^\perpF$ chosen with $\mu_{E_k^\perpF}^{(\delta m)_k}$.
\end{itemize}

The recursive, Markovian, nature of this definition ensures that the appropriate analog of Kolmogorov's consistency criterion holds automatically, so that it extends to a probability measure on $\underline{m}$-Matryoshkas in $\underline{F}$, i.e., to infinite sequences.

More precisely, letting $\algE_k$ denote the \salg\ generated by the random subspaces $E_0,\dots,E_k$ and $\sigma \big(\cup_{\kN} \algE_k\big)$ the \salg\ they generate, there is an unique $\sigma \big(\cup_{\kN} \algE_k\big)$-probability measure on $M$ which has the right finite-dimensional marginals. We complete this probability measure \big(with the \salg\ $\sigma \big(\cup_{\kN} \algE_k\big)$ suitably enlarged to a \salg\ we denote by $\algE$ to include null sets\big), leading to the complete filtered probability space $\big(M,\algE, (\algE_k)_{\kN},P\big)$.

It is clear from this construction that only $\underline{n}$ is relevant. The precise choice of $\underline{n}$-Matryo\-shka~$\underline{F}$, or when this comes into play of an orthonormal basis $(v_n)_{n\geq 1}$ in $\hilH$ with $F_k:=\psan (v_n)_{n\in \intt{1,n_k}}$ is irrelevant.

One basic question is under which conditions $\underline{E}$ is total \big(recall this means that $\cup_{\kN}E_k$ is dense in $\hilH$\big) with $P$-probability $1$. Denoting by $\overline{E}$ the closure of $\cup_{\kN}E_k$ in $\hilH$, the question is whether $\overline{E}^\perp$ \big(which is also the orthogonal of $\cup_{\kN}E_k$\big) reduces to $\{0\}$ (with $P$-probability $1$). This question has a simple and nice answer.

\begin{Theorem} \label{theo:maindense}
For $\kN^*$ define $r_k:=\frac{m_k-m_{k-1}}{n_{k}-m_{k-1}}$ if $n_{k}-m_{k-1}>0$ and $r_k:=1$ otherwise.

 If the series $\sum_k r_k$ is divergent the space $\overline{E}^\perp$ is zero-dimensional $($i.e., $\underline{E}$ is total$)$ with pro\-ba\-bility $1$.

 If the series $\sum_k r_k$ is convergent the space $\overline{E}^\perp$ is infinite-dimensional with probability $1$.
\end{Theorem}

This result is some kind of a paradox: if the series is convergent, even if the sequence $\underline{n}$ grows faster that the sequence $\underline{m}$ so that $\dim F_k-\dim E_k$ grows without bounds, the $F_k$s may succeed in filling the whole of $\hilH$! As a simple illustration, if $m_k=k$ and $n_k=2k$ for $\kN$ then $r_k=\frac{1}{k+1}$ and by the divergence of the harmonic series $\overline{E}^\perp=\{0\}$ with probability $1$ even if at each step $E_k^\perpF$ has dimension $k$, which grows without bounds.

This phenomenon is related to a classical paradox that we recall in the next section, explaining the name ``quantum $0-\infty$ law'' on the same occasion

\section{The Julie and Jack spending paradox} \label{sec:jjjparadox}

The Julie and Jack spending paradox is the following. Every month Julie and Jack receive $2$ gold coins each and they spend~$1$. They label the coins gathered on month $k$ as $2k-1$, $2k$ and make a pile of coins by putting the new earnings $2k-1$, then $2k$ on top of what remains of the previous earnings. However, they chose the coin they spend in a different way. Julie spends the coin at the top of the pile, while Jack spends the coin at the bottom of the pile. Thus, by the end of month $k$, they have received coins labeled $1,2,\dots,2k$ but Julie has spent coins $2,4,\dots,2k$ and kept coins $1,3,\dots,2k-1$ while Jack has spent coins $1,2,\dots,k$ and kept coins $k+1,k+2,\dots,2k$. Of course, each has~$k$ gold coins of savings, but letting~$k$ go to $\infty$, we see that by the end of times Julie has kept every coin with an odd label and is infinitely rich, while Jack has kept~\dots\ nothing and is ruined! We leave it to he reader to decide whether this is a~true paradox, or whether this simply teaches us something about the way set theory deals with infinity and limits of sets.

The Julie and Jack spending paradox is classical where classical can be given two meanings. First it is probably old, it's origin is obscure (to the author) and the above name is not canonical, but it is very widely known. Second, we shall exhibit a relation with our original problem which, we shall argue, is quantum, with a variant of the Julie and Jack spending paradox as a classical counterpart.

To make contact with our problem, we introduce a third person, John, who earns $2$ gold coins each month and spends $1$ just as Jack and Julie do. However, John is a gambler and chooses the coin he spends at random (uniformly) among the remaining ones. The question arises whether John is ruined by the end of times. As we shall soon explain, he is with probability $1$!

This paradox is related to the so-called Ross--Littlewood paradox (see for instance \cite[Example~$6a$]{Ross}, which also deals with the probabilistic side), which is usually presented as a supertask. A~supertask is a task consisting of a denumerable infinity of sub-taks, but which is performed in a finite amount of time. The name comes from computer science: one postulates the existence of a computer that performs a task in $1/2$ second, the next one in $1/4$ seconds, the next in $1/8$ seconds and so on, so that after one second a denumerable infinity of tasks, a supertask, has been accomplished. The possibility of such computers is highly speculative, but their theory is interesting and sheds some light on the limitations of standard computers. In the Ross--Littlewood paradox, the supertask is the following: an urn contains two balls at time 0, after $1/2$ seconds, one ball is removed and two are added, after another $1/4$ seconds one ball is removed and two are added, after another $1/8$ seconds one ball is removed and two are added, and so on. What remains in the urn at time $1$? If at each step the ball that is removed is one that was added just the step before, one has the analog of Julie's behavior, and at time $1$ the urn contains an infinite number of balls. However, if at each step the ball that is removed is chosen uniformly at random among the all the balls that are in the urn, one has an analog of John's behavior, and the urn is empty at time $1$.

We generalize the question, keeping the notations from Section \ref{sec:defmainres}. We assume that by the end of month $k$, Julie, Jack and John have earned a total of $n_k$ gold coins that they labeled from~$1$ to~$n_k$, and spent $m_k$ of them, with the same protocol as above. Each builds a pile of coins by putting the new earnings of month $k$, i.e., coins $n_{k-1}+1$ then $n_{k-1}+2$ up to $n_k$ on top of the what remains of the previous ones. The description of the expenses of Julie, who takes the coins she spends on the top of her pile, is a bit involved but irrelevant for our discussion and we leave it to the reader. Jack has spent coins labeled from $1$ to $m_k$. By hypothesis, $\lim_{k\to \infty} m_k=\infty$ so the fate of Jack is always complete ruin, independently of the growth of the sequence $\underline{n}$. For the expenses of John we give a recursive definition. By the end of month $1$, John has spent the coins labeled by a set $S_1$ of size $m_1$ chosen uniformly at random in $\intt{1,n_1}$. By the end of month $2$, he has spent the coins labeled by a set $S_2 \supset S_1$ of size $m_2$ chosen uniformly at random in $\intt{1,n_2}$, and so on. Choosing $S_{k+1} \supset S_k$ of size $m_{k+1}$ uniformly at random in $\intt{1,n_{k+1}}$ is the same as choosing a subset of size $(\delta m)_k:=m_{k+1}-m_k$ uniformly at random in $\intt{1,n_{k+1}}\setminus S_k$ which has size $n_{k+1}-m_k$. There are $\binom{n_{k+1}-m_k}{m_{k+1}-m_k}$ such subsets, so each has probability $\binom{n_{k+1}-m_k}{m_{k+1}-m_k}^{-1}$. This description extends to a probability measure on infinite sequences $\underline{S}:=(S_k)_{\kN}$ of subsets of~$\mathbb{N}$ with $\varnothing:=S_0\subset S_1 \subset S_2 \subset\cdots$, where $S_k\subset \intt{1,n_k}$ and $\# S_k =m_k$ for $\kN$. The generalized question is whether John is ruined or not by the end of times. This is answered by the following result:

\begin{Theorem} \label{theo:mainfill}
For $\kN^*$ define $r_k:=\frac{m_k-m_{k-1}}{n_{k}-m_{k-1}}$ if $n_{k}-m_{k-1}>0$ and $r_k:=1$ otherwise.
\begin{enumerate}\itemsep=0pt

\item[] If the series $\sum_k r_k$ is divergent then $\cup_{\kN} S_k=\mathbb{N}^*$ with probability $1$.

\item[] If the series $\sum_k r_k$ is convergent then $\mathbb{N}^* \setminus \cup_{\kN} S_k$ is infinite with probability $1$.
 \end{enumerate}
\end{Theorem}

We note the complete analogy with Theorem \ref{theo:maindense}. The original paradox is when $m_k=k$ and $n_k=2k$ for $\kN$, so $r_k=\frac{1}{k+1}$ and by the end of times, John is ruined with probability $1$. Note also that John's fate has some kind of dual relationship with the Birthday paradox and the Coupon collector problem.

If an orthonormal basis $(v_n)_{n\geq 1}$ in $\hilH$ is fixed and $F_k:=\psan (v_n)_{n\in \intt{1,n_k}}$, there is another simple way to put a measure, say $P_{\underline{F},\underline{m}}^{\text{classical}}$, on $\underline{m}$-Matryoshkas in $\underline{F}$, namely by choosing $\underline{S}$ as above and taking $E_k=\psan (v_n)_{n\in S_k}$.

To explain why we consider the measure $P_{\underline{F},\underline{m}}^{\text{classical}}$ as classical whereas $P_{\underline{F},\underline{m}}$ \big(which we could write $P_{\underline{F},\underline{m}}^{\text{quantum}}$\big) is of quantum nature, we take the analogy with quantum computing. A classical bit, or Cbit, can be in only two states, $0$ and $1$, and there are only two reversible operations on a Cbit, doing nothing or permuting $0$ and $1$. A quantum bit, or Qbit, on the other hand, is a unit vector in a two-dimensional Hilbert space, it can be expressed as a linear combination of two basis vectors, but every unitary transformation represents a reversible operation. These considerations generalize immediately to a collection of $n>1$ Cbits or Qbits. For Cbits, there are $2^n$ states and a reversible transformation permutes those states. For Qbits the states span a Hilbert space of dimension $2^n$ and a reversible transformation is an unitary on that space. In exactly the same sense, $P_{\underline{F},\underline{m}}^{\text{classical}}$ involves no superposition of states and randomness involves uniform choices with respect to permutations, which are classical symmetries, whereas $P_{\underline{F},\underline{m}}^{\text{quantum}}$ involves superposition of states, a quantum operation, and randomness involves uniform choices with respect to unitary transformations, which are quantum symmetries.

However, the outcome is the same: the alternative for the status of $\cup_{\kN}E_k$ within $\hilH$ is the same under the classical and the quantum constructions.

This classical/quantum distinction can also be rephrased as an illustration of the nuance between combinatorics and geometric probability (sometimes also called continuous combinatorics). Cases when the two yield strikingly similar results abound, and the similarity between
Theorems \ref{theo:maindense} and \ref{theo:mainfill} is but another illustration.

We shall not give (reproduce?) a detailed proof of Theorem \ref{theo:mainfill} because it is simpler than the proof of Theorem \ref{theo:maindense} but follows closely the same lines.\footnote{In fact, the main differences in the proof are really due to the classical/quantum difference: classical discrete $0-1$ \rv s which in terms of subspaces mean orthogonal or included, become continuous, with subspaces almost orthogonal or close to each other.} The heart of the similarity between the classical and the quantum case is embodied in the following elementary fact: if $n\in \intt{n_{i-1}+1,n_i}$ for some $\iN^*$, and $k\geq i$, the probability that $n$ does not belong to $S_k$ is $\prod_{i}^{k} (1-r_j)$. Just as its quantum counterpart given in Corollary \ref{cor:prodineq}, this is proven by conditioning and recursion. The identity says that certain covariances among random unit vectors are the same in the classical and the quantum case, but this does not survive for higher correlations.

\section{Main proofs} \label{sec:mainproofs}

Fix an $\underline{m}$-Matryoshka $\underline{E}$ and let $v\in \hilH$. For $\kN$ denote by $U_k(v)$ (or simply $U_k$ when there is no risk of confusion) the orthogonal projection of $v$ on $E_k$ (a finite-dimensional, hence closed, subspace of $\hilH$) and set $W_k=W_k(v):=v-U_k(v)$, the projection of $v$ on the orthogonal complement of $E_k$ in $\hilH$. Set $V_0(v):=0$ and for $\kN^*$ set $V_k=V_k(v):=U_k(v)-U_{k-1}(v)$, which by construction is the orthogonal projection of $v$ on the orthogonal complement of $E_{k-1}$ in $E_k$. To summarize, for $\kN^*$, $v\in \hilH$ splits as a sum of three mutually orthogonal vectors
\[
v=U_{k-1}(v)+V_k(v)+W_k(v)
\]
corresponding to the orthogonal sum decomposition
\[
 \hilH = E_{k-1} \oplus \big(E_{k-1}^\perp  \cap  E_k\big) \oplus E_k^\perp.
 \]
Below is a picturesque representation of this decomposition, where each axis represents a vector space:

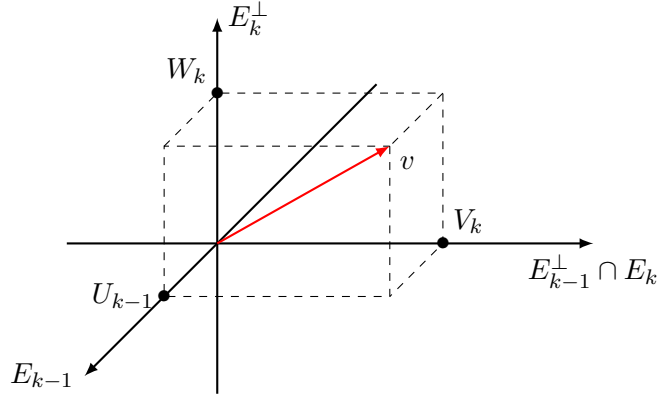
\tikzset{math3d/.style=
 {x= {(-0.353cm,-0.353cm)}, z={(0cm,1cm)},y={(1cm,0cm)}}}
\tikzmath{\u = 2; \v =3; \w=2; \l=2;}
$$
\begin{tikzpicture}[math3d, >=latex]
% Axes
\draw[thick, ->] (-3*\l,0,0) -- (2.5*\l,0,0) node[left] {$E_{k-1}$};
\draw[thick, ->] (0,-\l,0) -- (0,2.5*\l,0) node[below] {$E_{k-1}^\perp \cap E_k$};
\draw[thick, ->] (0,0,-\l) -- (0,0,1.5*\l) node[right] {$E_k^\perp$};
% Useful points
\coordinate (O) at (0,0,0);
\coordinate (U) at (\u,0,0);
\coordinate (V) at (0,\v,0);
\coordinate (W) at (0,0,\w);
\coordinate (UV) at (\u,\v,0);
\coordinate (UW) at (\u,0,\w);
\coordinate (VW) at (0,\v,\w);
\coordinate (v) at (\u,\v,\w);
% The cube and the arrow
\draw[dashed] (v) -- (UV) -- (U) -- (UW) -- cycle;
\draw[dashed] (UW) -- (W) -- (VW) -- (V) -- (UV);
\draw[dashed] (VW) -- (v);
\draw[red, thick, ->] (O) -- (v);
\draw (U) node[left] {$U_{k-1}$} node{$\bullet$};
\draw (V) node[above right] {$V_k$} node{$\bullet$};
\draw (W) node[above left] {$W_k$} node{$\bullet$};
\draw (v) node[below right] {$v$};
\end{tikzpicture}
$$
We collect a few simple facts in two lemmas:

\begin{Lemma} \label{lem:basicfacts0} Fix an $\underline{m}$-Matryoshka $\underline{E}$ and let $v\in \hilH$.
\begin{enumerate}\itemsep=0pt
 \item[$(i)$] For $\kN^*$, $v=W_k(v)+\sum_{j=1}^k V_j(v)$ is an orthogonal decomposition.
 \item[$(ii)$] The series $\sum_j V_j(v)$ converges in $\hilH$ and so does the sequence $W_k(v)$.
 \item[$(iii)$] The limit $\lim_{k\to \infty} W_k(v)$ is the orthogonal projection of $v$ on $\overline{E}^\perp$.
\end{enumerate}
\end{Lemma}

\begin{proof}
The decomposition of $v$ follows immediately from the definitions, and its orthogonality is a consequence of the fact that $(E_k)_{\kN}$ is an increasing sequence of subspaces and $V_j(v)$ belongs by definition to the orthogonal complement of $E_{j-1}$ in $E_j$ for $\jN^*$, while $W_k(v)$ is orthogonal to $E_k$, establishing $(i)$.

The convergence of the series $\sum_j V_j(v)$ is an immediate consequence of the orthogonality of the decomposition in $(i)$ by a standard Hilbert space argument,\footnotemark \
establishing $(ii)$.

Recall that $\overline{E}$ is the closure of $\cup_{\kN}E_k$ in $\hilH$. If $v\in \overline{E}^\perp$, then $W_k(v)=v$ for every $\kN$. If $v\in \cup_{\kN}E_k$, then $W_k(v)=0$ for large enough $k$. If~$v\in \overline{E}$, for every $\varepsilon >0$ there is a $\jN$, and~$v'$, $v''$ with $v'\in E_j$, $v''\in E_j^\perp$ such that $v=v'+v''$ and $\norm{v''}\leq \varepsilon$. Then $\norm{W_k(v'')}\leq \varepsilon$ for every $\kN$ while for $k\geq j$, $W_k(v')=0$ so $\norm{W_k(v)}\leq \varepsilon$ for $k\geq j$. This finishes the proof of~$(iii)$.
\end{proof}
\footnotetext{For the sake of completeness: as the summands $V_j(v)$ are orthogonal, $\norm{\sum_{j=1}^k V_j(v)}^2=\sum_{j=1}^k \norm{V_j(v)}^2 =\norm{v}^2-\norm{W_k(v)}^2 \leq \norm{v}^2$, so that $\sum_{j\geq 1} \norm{V_j(v)}^2 < +\infty$, i.e., the series with general term $\norm{V_j(v)}^2$ is a Cauchy sequence and then so is the series with general term $V_j(v)$.}

Recall that $\big(M,\algE, (\algE_k)_{\kN},P\big)$ is the complete filtered probability space we work with, where $M:=M_{\underline{F},\underline{m}}$ denotes the set of $\underline{m}$-Matryoshkas in $\underline{F}$, $\algE_k$ denotes, for $\kN$, the \salg\ generated by $k$ first components of a Matryoshka in $M$, $\algE$ contains $\sigma \big(\cup_{\kN} \algE_k\big)$, and $P:=P_{\underline{F},\underline{m}}$ is a complete $\algE$-probability measure defined recursively as explained briefly in Section \ref{sec:defmainres}.

\begin{Lemma} \label{lem:basicfacts1} Let $v\in \hilH$.
\begin{enumerate}\itemsep=0pt
\item[$(i)$] As functions on $(M,\algE)$ with values in $\hilH$, $U_k(v)$, $V_k(v)$ and $W_k(v)$ are $\algE_k$-measurable for each $\kN$.

\item[$(ii)$] The sequence $W_k(v)$ converges almost surely to $0$ if and only if
\[
\lim_{k\to \infty} \expec{\norm{W_k(v)}^2}=0.
\]
\end{enumerate}
\end{Lemma}

\begin{proof}
As $\algE_k$ is the \salg\ generated by the random subspaces $E_0,\dots,E_k$ of a Matryoshka, $E_k$ is $\algE_k$-measurable by definition, so the projection $U_k(v)$ of $v$ on $E_k$ is $\algE_k$-measurable, and then so are $W_k(v)=v-U_k(v)$ and $V_k(v)=U_k(v)-U_{k-1}(v)$ (or $0$ if $k=0$), establishing $(i)$.

By $(ii)$ in Lemma \ref{lem:basicfacts0} and continuity of the norm, $Q(v):=\lim_{k\to \infty}\norm{W_k(v)}^2$ exists pointwise. In fact $\norm{W_k(v)}^2$ decreases as a function of $k$, pointwise, so by dominated convergence (for instance) $\lim_{k\to \infty} \expec{\norm{W_k}^2}$ exists and equals $\expec{Q}$. Hence if $W_k(v)$ converges almost surely to $0$, equivalently if $\norm{W_k(v)}^2$ converges almost surely to $0$, then $\lim_{k\to \infty} \expec{\norm{W_k(v)}^2}=0$, proving the direct implication of $(ii)$. Conversely, if $\lim_{k\to \infty} \expec{\norm{W_k(v)}^2}=0$ then $\expec{Q}=0$ which as $Q\geq 0$ implies that $Q=0$ almost surely, i.e., $\lim_{k\to \infty} W_k=0$ almost surely.
\end{proof}

Recall that by assumption $\cup_{\iN}F_i$ is dense in $\hilH$. So a subspace of $\hilH$ is dense if and only if the closure of that subspace contains $F_i$ for every $\iN^*$. Thus Lemma~\ref{lem:basicfacts1} embodies a strategy of proof of Theorem~\ref{theo:maindense} based on the study of $\expec{\norm{W_k(v)}^2}$ at large $k$ for $v\in \cup_{\iN}F_i$.

Suppose that $v\in F_i$ for some $\iN^*$. Then $v\in F_k$ for $k\geq i$, and then as the $E_j$s, $j\leq k$ are subspaces of $F_k$, each component in the decompositions $v=U_k(v)+W_k(v)=U_{k-1}(v)+V_k(v)+W_k(v)=W_k(v)+\sum_{j=1}^k V_j(v)$ belong to~$F_k$.

We shall use the following lemma:

\begin{Lemma}
Let $F$ be a finite-dimensional Hilbert space and $E'\subset F'\subset F$ be subspaces, whose dimensions are denoted by $m'$, $n'$, $n$. Let $v\in F'$. If $m'\leq m\leq n$, let $E$ be a random subspace of~$F$ of dimension $m$ containing $E'$ $($chosen in a unitary invariant way$)$ and let $E^{\perp}$ be its orthogonal complement in $F$. Write $v=u+V+W$ as an orthogonal decomposition, where $u\in E'$, $V$ belongs to the orthogonal complement of $E'$ in $E$, and $W \in E^{\perp}$.

Then $\expec{\norm{W}^2}=\norm{v-u}^2\frac{n-m}{n-m'}$ if $n>m'$ and $0$ otherwise.
\end{Lemma}

\begin{proof}
If $m'=n$ then $m'=m$, $E'=E=F'=F$ so $v=u$, $W=0$ and $\expec{\norm{W}^2}=0$. If~$n>m'$, the result is simply the translation of the basic property of the Grassmannian measure recalled above for the vector $v-u$, which lays down in the orthogonal complement of~$E'$ in~$F$ (of dimension $n-m'$) with $W$ its projection on $E^{\perp}$, a random subspace of dimen\-sion~$n-m$.\looseness=1
\end{proof}

We can now make use the Markov property of the measure on Matryoshkas. Recall that, for $\kN^*$, $r_k:=\frac{m_k-m_{k-1}}{n_{k}-m_{k-1}}$ if $n_{k}-m_{k-1}>0$ and $r_k:=1$ otherwise.

\begin{Lemma} \label{lem:markovMatryoshkas} Let $\iN^*$ and $v\in F_i$.
 For $k\geq i$, $\expec{\norm{W_k(v)}^2|\algE_{k-1}}=\norm{W_{k-1}(v)}^2(1-r_k)$.
\end{Lemma}

\begin{proof}
Write $v=U_{k-1}+V_k+W_k$. Then apriori $U_{k-1}$, $V_k$ and $W_k$ are random variables, but conditionally on $\algE_{k-1}$ we are in position to apply the previous lemma and get $\expec{\norm{W_k}^2|\algE_{k-1}}=\norm{v-U_{k-1}}^2\frac{n_k-m_k}{n_k-m_{k-1}}$ if $n_{k}-m_{k-1}>0$ and $0$ otherwise. By construction, $v-U_{k-1}=W_{k-1}$ so using the definition of $r_k$ gives the announced formula.
\end{proof}

There is a useful corollary:

\begin{Corollary} \label{cor:prodineq} Let $\iN^*$ and $v\in F_i$.
For $k \geq i$,
\[
\expec{\norm{W_k(v)}^2}=\expec{\norm{W_{i-1}(v)}^2}\prod_{j=i}^{k} (1-r_j) \leq \norm{v}^2\prod_{j=i}^{k} (1-r_j),
\]
with equality if $v$ belongs to the orthogonal complement of $F_{i-1}$ in $F_i$.
\end{Corollary}

\begin{proof}
{\sloppy
For $k>0$, we know that $v-U_{k-1}=W_{k-1}$ the formula in Lemma \ref{lem:markovMatryoshkas} yields $\expec{\norm{W_k}^2|\algE_{k-1}}\allowbreak=\norm{W_{k-1}}^2(1-r_k)$, which can be used recursively by nesting of conditional expectations to yield $\expec{\norm{W_k}^2}=\expec{\norm{W_{i-1}}^2}\prod_{j=i}^{k} (1-r_j)$.\footnotemark

}

The inequality follows from the orthogonality properties in the decompositions $v=U_{i-1} + V_i + W_i$, equivalently $v-U_{i-1}= V_i + W_i$, which yield $\norm{v}^2= \norm{U_{i-1}}^2+\norm{V_i(v)}^2+\norm{W_i}^2 \geq \norm{V_i}^2+\norm{W_i}^2=\norm{v-U_{i-1}}^2=\norm{W_{i-1}}^2$.

Finally, if $v$ belongs to the orthogonal complement of $F_{i-1}$ in $F_i$ then $v$ is orthogonal to $E_{i-1} \subset F_{i-1}$, i.e., $U_{i-1}=0$ and the inequality turns into an equality.
\end{proof}
\footnotetext{Note that one cannot go further down because $W_{i-1}$ has no reason apriori to belong to $F_{i-1}$.}

We can now prove the first assertion of Theorem \ref{theo:maindense}.

\begin{proof}[Proof the first assertion of Theorem \ref{theo:maindense}]
By assumption, $\cup_{\iN}F_i$ is dense in $\hilH$. So $\cup_{\kN}E_k$ is dense in $\hilH$ if and only if it is dense in $\cup_{\iN}F_i$. Equivalently, $\cup_{\kN}E_k$ is dense in~$\hilH$ if and only if any $v$ in some $F_i$ is in the closure of $\cup_{\kN}E_k$, i.e., if and only if the sequence $(W_k(v))_{k\geq i}$ converges to $0$ for any $v$ in some $F_i$.

Using the linearity of the map $v\mapsto W_k(v)$, to prove that $\lim_{k\to \infty} W_k(v)=0$ for every $v\in \cup_{\iN}F_i$, it suffices to check this convergence when $v$ ranges over a basis of the finite-dimensional space $F_i$ for each $i$. We fix an orthonormal basis $(v_n)_{n\geq 1}$ in $\hilH$ such that $F_i:=\psan (v_n)_{n\in \intt{1,n_i}}$. Then $W_k(v_n))$ is well-defined for large enough $k$, i.e., for $k$ such that $n_k\geq n$: $\cup_{\kN}E_k$ is dense in $\hilH$ if and only if $\lim_{k \to \infty} W_k(v_n)=0$ for every $\nN^*$.

By $(ii)$ in Lemma \ref{lem:basicfacts1} we know that $\lim_{k\to \infty} W_k(v)=0$ holds almost surely for a given $v$ in some~$F_i$ if and only if $\lim_{k\to \infty} \expec{\norm{W_k(v)}^2}=0$. By the above remark, it is enough to let $v$ range over the countable set $(v_n)_{n\geq 1}$.

By Corollary \ref{cor:prodineq}, for $\iN^*$, $n\in \intt{n_{i-1}+1,n_i}$ and $k\geq i$, $\expec{\norm{W_k(v_n)}^2}= \prod_{i}^{k} (1-r_j)$. Thus $\cup_{\kN}E_k$ is dense in $\hilH$ if and only if the infinite product ``diverges'' to~$0$, i.e., if $\lim_{k\to \infty} \prod_{i}^{k} (1-r_j)\allowbreak=0$ for every $\iN^*$. As $r_j\in [0,1]$ for each $\jN ^*$, this is equivalent to the divergence of the series $\sum_j r_j$, concluding the proof.
\end{proof}

For the second assertion, we shall make use of the following general but elementary Hilbert space theory standard fact:

\begin{Lemma} \label{lem:symnormproj}
 Let $G$, $H$ be two closed subspaces of $\hilH$ and let $\Phi$, $\Psi$ denote the corresponding orthogonal projectors. Then the norm of $\Phi$ restricted to $H$ and the norm of $\Psi$ restricted to $G$ are equal.
\end{Lemma}

\begin{proof}
Recall that $\aver{\; , \;}$ denotes the inner product in~$\hilH$. For every $v\in G$ and $w\in H$, $\aver{v,\Phi w}=\aver{\Psi v,w}$ (because both equal $\aver{v,w}$). By the Cauchy--Schwarz inequality $|\aver{v,\Phi w}|=|\aver{\Psi v,w}| \leq \norm{\Psi v} \norm{w}$. If $\delta\geq 0$ is such that $\norm{\Psi v}\leq \delta \norm{v}$ then $|\aver{v,\Phi w}|\leq \delta \norm{v} \norm{w}$. If~$\norm{\Psi v}\leq \delta \norm{v}$ holds for every $v\in G$, take $v=\Phi w$ to obtain $\norm{\Phi w}^2=|\aver{\Phi w,\Phi w}|\leq \delta \norm{\Phi w} \norm{w}$ for every $w\in H$. If~$\Phi w\neq 0$ divide by $\norm{\Phi w}$ to get $\norm{\Phi w}\leq \delta \norm{w}$, and this inequality is also true if $\Phi w= 0$. Thus the norm of $\Phi$ restricted to $H$ is at most the norm of $\Psi$ restricted to $G$. By symmetry, they are equal.
\end{proof}

\begin{Lemma} \label{lem:densecharac}
Let $(H_k)_{\kN}$ be an increasing sequence of subspaces of $\hilH$. Then $\cup_{\kN} H_k$ is dense in~$\hilH$ if and only if for every finite-dimensional subspace $G$ of $\hilH$ and every $\varepsilon >0$ there is a~$\kN$ such that the orthogonal projection on $G$ restricted to the orthogonal complement of $H_k$ in $\hilH$ has norm $\leq \varepsilon$.
\end{Lemma}

\begin{proof}
For $\kN$, let $H_k^{\perp}$ denote the orthogonal complement of $H_k$ in $\hilH$ and $\Psi_k$ the orthogonal projection on $\big(H_k^{\perp}\big)^{\perp}$, which is the closure of $H_k$. Set $\Psi_k^{\perp}:=\Id -\Psi_k$, the orthogonal projection on $H_k^{\perp}$.

The ``if'' part is the easiest. Fix $v$ in $\hilH$. We want to approach $v$ arbitrarily closely by vectors in $\cup_{\kN} H_k$. The case $v=0$ is trivial and we assume $v\neq 0$. Let $G=\psan (v)$. By hypothesis, for every $\varepsilon >0$ there is a $k$ such that the orthogonal projection on $G$ restricted $H_k^{\perp}$ has norm $\leq \frac{\varepsilon}{2\norm{v}}$. As $\Id -\Psi_k$ is the orthogonal projector on $H_k^{\perp}$, by Lemma \ref{lem:symnormproj}, $\norm{(\Id -\Psi_k)v} \leq \frac{\varepsilon}{2\norm{v}}\norm{v}=\frac{\varepsilon}{2}$. As~$H_k$ is dense in its closure, which contains $\Psi_kv$, one can find $w\in H_k$ such that $\norm{\Psi_kv-w}\leq \frac{\varepsilon}{2}$. By the triangular inequality, $\norm{v-w} \leq \varepsilon$. To summarize, for an arbitrary $v\neq 0$ in $\hilH$ and an arbitrary $\varepsilon >0$ there is a $\kN$ and a $w\in H_k$ such that $\norm{v-w} \leq \varepsilon$. Hence $\cup_{\kN} H_k$ is dense in~$\hilH$.\looseness=-1

For the ``only'' if part, let $G$ be a finite-dimensional subspace of $\hilH$, say of dimension $d$. We~may assume that $d>0$. Pick an orthonormal basis $(v_1,\dots,v_d)$. Let $\varepsilon >0$. As $\cup_{\kN} H_k$ is dense in $\hilH$, for each $\alpha\in \intt{1,d}$ there is $k_\alpha\in \mathbb{N}$ and $w_\alpha\in H_{k_\alpha}$ such that $\norm{v_\alpha-w_\alpha}\leq \varepsilon d^{-1/2}$. Set $k:=\max_{\alpha} k_\alpha$. As $(H_l)_{\lN}$ is an increasing sequence, $w_\alpha\in H_k$ for each $\alpha\in \intt{1,d}$. The orthogonal projection is a closest point so $\norm{\Psi_k^{\perp}v_\alpha}=\norm{(\Id -\Psi_k)v_\alpha}\leq \norm{v_\alpha-w_\alpha}\leq \varepsilon d^{-1/2}$. Write an arbitrary $v \in G$ as $v:=\sum_\alpha \lambda_{\alpha} v_\alpha$ for a $(\lambda_{\alpha})_{\alpha\in \intt{1,d}} \in \mathbb{C}^d$. Then
\begin{align*}
\norm{\Psi_k^{\perp}v}^2 & = \sum_{\alpha,\beta \in \intt{1,d}}\overline{\lambda_{\alpha}} \lambda_{\beta} \big\langle \Psi_k^{\perp}v_\alpha,\Psi_k^{\perp}v_\beta\big\rangle
\\
& \leq \sum_{\alpha,\beta \in \intt{1,d}} |\lambda_{\alpha}|\,|\lambda_{\beta}|\norm{\Psi_k^{\perp}v_\alpha} |\norm{\Psi_k^{\perp}v_\beta}
\\
& \leq \varepsilon^2 d^{-1} \bigg(\sum_{\alpha\in \intt{1,d}} |\lambda_{\alpha}|\bigg)^2
\\
& \leq \varepsilon^2 d^{-1} \bigg(d\sum_{\alpha\in \intt{1,d}} |\lambda_{\alpha}|^2\bigg)=\varepsilon^2 \norm{v}^2 ,
\end{align*}
using different avatars of the Cauchy--Schwarz inequality when going from the fist line to the second, and then from the third line to the fourth. We have proven that for every $\varepsilon >0$ there is a $\kN$ such that the orthogonal projection on $H_k^{\perp}$ restricted to~$G$ has norm $\leq \varepsilon$. Noting that~$G$ is finite-dimensional, hence closed, we may apply Lemma~\ref{lem:symnormproj} with $H:=H_k^{\perp}$, concluding the proof.\looseness=-1
\end{proof}

If $\underline{E}$ is an $\underline{m}$-Matryoshka, let $G_{\underline{E}}$, or simply $G$ when no confusion is possible, denote the orthogonal complement of $\cup_{\kN}E_k$, always a closed subspace of $\hilH$. Let $\Phi_{\underline{E}}$, or simply $\Phi$ when no confusion is possible, denote the orthogonal projector on $G_{\underline{E}}$.

Set $A:=\big\{\underline{E} \in M_{\underline{F},\underline{m}} \text{ such that } G_{\underline{E}}\text{ is finite-dimensional}\big\}$. Our aim is to show that $A$ is a~null set, and as we work in a complete probability space, it is enough to show that it is contained in a measurable set of measure $0$.

For $\varepsilon >0$ and $\kN$, let $A_{k}(\varepsilon):=\big\{\underline{E} \in M_{\underline{F},\underline{m}} \text{ such that } \norm{\Phi_{\underline{E}}v} \leq \varepsilon \norm{v} \text{ for every } v \in F_k^{\perp}\big\}$.

\begin{Lemma}% \label{lem:AsubsetcupAepsk}
Each $A_{k}(\varepsilon)$ is measurable. For given $\varepsilon$ the sets $A_{k}(\varepsilon)$ increase with $\kN$, and $A\subset \cup_{\kN} A_{k}(\varepsilon)=\lim_{k\to \infty} A_{k}(\varepsilon)$ for every $\varepsilon >0$.
\end{Lemma}

\begin{proof} Fix $\varepsilon >0$.
 If $\underline{E}$ is any $\underline{m}$-Matryoshka, recall that, for $v\in \hilH$ and $\kN$, $W_k(v)$ is the orthogonal projection of $w$ on the orthogonal complement of $E_k$ in $\hilH$. By $(i)$ in Lemma~\ref{lem:basicfacts1}, the map $\underline{E} \in (M,\algE) \mapsto W_k(v)$ is $\algE_k$-measurable, hence $\algE$-measurable, for each $\kN$ and each $v\in \hilH$. The content of $(iii)$ in Lemma~\ref{lem:basicfacts0} is that $\Phi_{\underline{E}} v=\lim_{k\to \infty} W_k(v)$, so $\underline{E} \in (M,\algE) \mapsto \Phi_{\underline{E}}(v)$ is measurable as a limit of measurable maps. Thus, for given $v\in \hilH$, the set $A_v:=\big\{\underline{E} \in M_{\underline{F},\underline{m}} \text{ such that } \norm{\Phi_{\underline{E}}v}\leq \varepsilon \norm{v} \big\}$ is measurable. The orthogonal projection $\Phi_{\underline{E}}$ is continuous and $\hilH$ is separable. Hence, if $H$ is any closed subspace of $H$, and $(v_n)_\nN$ a~dense sequence of vectors in $H$, the set $\big\{\underline{E} \in M_{\underline{F},\underline{m}} \text{ such that } \norm{\Phi_{\underline{E}}v} \leq \varepsilon \norm{v} \text{for every } v \in H\big\}$ is measurable because it can be written as the countable intersection $\cap_{\nN} A_{v_n}$. Applying this result to $H= F_k^{\perp}$ gives the measurability of each $A_{k}(\varepsilon)$.

As $(F_k)_{\kN}$ is an increasing sequence of subspaces of $\hilH$, $(F_k^\perp)_{\kN}$ is a decreasing sequence of subspaces of $\hilH$, so the sets $A_{k}(\varepsilon)$ increase with $\kN$ for fixed $\varepsilon$.

Recall that $\cup_{\kN} F_k$ is dense in $\hilH$. Applying Lemma \ref{lem:densecharac} to $H_k=F_k$ for $\kN$ yields that $A$ is contained in $\cup_{\kN} A_{k}(\varepsilon)$, which is also $\lim_{k\to \infty} A_{k}(\varepsilon)$ by the previous argument.
\end{proof}

We can now turn to the heart of the matter.

\begin{proof}[Proof the second assertion of Theorem \ref{theo:maindense}]
Our assumption now is that $\sum_k r_k$ is convergent. Hence, for large enough $i$, $s_i:=\prod_{k\geq i} (1-r_k) >0$, and $s_i \nearrow 1$ when $i\to \infty$.

By assumption, $\lim_{i\to \infty} n_i=\infty$. Thus $\mathbb{S}:=\{\iN^*, n_i-n_{i-1} >0\}$ is infinite. For each $i\in \mathbb{S}$ we choose a non-zero vector $v_i\in F_{i-1}^{\perp} \cap F_i$. By Corollary \ref{cor:prodineq}, $\expec{\norm{W_k(v_i)}^2}=\norm{v_i}^2\prod_{j=i}^{k} (1-r_j)$ for $k\geq i$. Taking the large $k$ limit, $\expec{\norm{\Phi v_i}^2}=\norm{v_i}^2s_i$ for $i\in \mathbb{S}$.

Fix $\varepsilon >0$, $\kN$ and $v\in F_k^\perp$. Write
\[ \expec{\norm{\Phi v}^2}=\expec{\norm{\Phi v}^2\ind{A_{k}(\varepsilon)}}+ \expec{\norm{\Phi v}^2\ind{A_{k}(\varepsilon)^c}}.\]
On $A_{k}(\varepsilon)$, $\norm{\Phi v}^2\leq \varepsilon^2 \norm{v}^2$ so
\[ \expec{\norm{\Phi v}^2\ind{A_{k}(\varepsilon)}}\leq \varepsilon^2 \norm{v}^2 P(A_{k}(\varepsilon)).\]
As $\norm{\Phi v}^2 \leq \norm{v}^2$ ($\Phi$ is an orthogonal projector!)
\[ \expec{\norm{\Phi v}^2\ind{A_{k}(\varepsilon)^c}}\leq \norm{v}^2P\big(A_{k}(\varepsilon)^c\big).\]

Thus we have proven that, for $\varepsilon >0$, $\kN$ and $v\in F_k^\perp$,
\[ \expec{\norm{\Phi v}^2} \leq \varepsilon^2 \norm{v}^2 P(A_{k}(\varepsilon)) +\norm{v}^2(1-P(A_{k}(\varepsilon)).\]
For $i\in \mathbb{S}$, $i>k$, it holds that $v_i\in F_{i-1}^{\perp} \cap F_i \subset F_k^\perp$, so putting things together
\[ \norm{v_i}^2s_i=\expec{\norm{\Phi v_i}^2} \leq \big(\varepsilon^2-1\big) \norm{v_i}^2 P(A_{k}(\varepsilon)) +\norm{v_i}^2.\]
Dividing by $ \norm{v_i}^2$ (which is $>0$ by construction) and rearranging yields
\[\big(1-\varepsilon^2\big)P(A_{k}(\varepsilon))\leq 1-s_i,\]
and letting $i\to \infty$ along $\mathbb{S}$ we obtain that $P(A_{k}(\varepsilon))=0$ for every $\varepsilon <1$, so $\cup_{\kN} A_{k}(\varepsilon)=\lim_{k\to \infty} A_{k}(\varepsilon)$ is a null set for every $\varepsilon <1$.

Thus $A:=\big\{\underline{E} \in M_{\underline{F},\underline{m}} \text{ such that } G_{\underline{E}}\text{ is finite-dimensional}\big\}$, a subset of $\cup_{\kN} A_{k}(\varepsilon)$ is also a null set. To summarize, if
$\sum_k r_k$ is convergent, the orthogonal $G_{\underline{E}}$ of $\cup_{\kN}E_k$ is infinite-dimensional with probability $1$, concluding the proof.
\end{proof}

\section{Conclusions} \label{sec:concl}

The conclusions of Theorem \ref{theo:maindense} give another example in probability theory when in fact the outcome of randomness is certitude, and depending on a condition (the convergence or divergence of a series) only two extreme alternatives are possible.

However, the above mathematical considerations, neat as they are, should be seen as no more than a preliminary to more interesting physical questions: once one is guaranteed that the full Hilbert space is ``covered'', some generic features of interest can in principle be studied.

To fix the ideas, let us assume that $m_k=k$ for $\kN$, i.e., at each step a new ray in $\hilH$ is added, constructing step by step a random orthonormal basis (modulo phases) $(u_k)_{k\geq 1}$, which is to be compared to the original orthonormal basis $(v_l)_{l\geq 1}$ in $\hilH$ such that $F_k:=\psan (v_l)_{l\in \intt{1,n_k}}$. More precisely, we would like to understand the behavior of matrix elements in the $u$-basis of operators which are local in the $v$-basis. To really make contact with, say, the eigenstate thermalization hypothesis, we would also need a Hamiltonian, providing an increasing sequence $(\omega_k)_{k\geq 1}$ of energies such that such $u_k$ has energy $\omega_k$.

Up to now, we have said nothing neither about the choice of sequence of dimensions of the trial spaces $F_k$, nor about the energies. As a conclusion, we propose a model of random Hamiltonians in which both sequences are proportional.

We thus choose a random integer sequence $n_k=:\dim F_k$, an energy scale $\omega$ and set $\omega_k=\omega \dim F_k$. The energy scale itself could be deterministic or random. It would be nice to have some physical intuition for such a proportionality relation, which is admittedly artificial.\footnote{The relationship with the more physical feature that the Hamiltonian should be somewhat local (though not strictly diagonal) in the original orthonormal basis $(v_l)_{l\geq 1}$ is obscure to say the least. Interpreted in terms of trial wave functions, the proportionality suggests that the larger the gap between two states, the more one should enlarge the trial space, a principle which looks questionable as well.}

Our proposal is based on the use of renewal processes. We fix a sequence $(p_n)_{\nN ^*}$, where each $p_n$ is $\geq 0$ and $\sum_{\nN ^*} p_n=1$, and use it to define a renewal process $(D_k)_{\kN}$: $D_0=0$ and the differences $D_k - D_{k-1}$, $\kN ^*$, are integer valued independent random variables taking value $\nN ^*$ with probability $p_n$. Taking $\dim F_k:=D_k$ for $\kN$ and carrying out the construction, the strong laws of large numbers lead to criteria to ensure that $(u_k)_{k\geq 1}$ is indeed total (with probability $1$ with respect to the renewal process). This is the case obviously if (but not only~if) the distribution $(p_n)_{\nN ^*}$ has a mean i.e., if $\sum_{\nN ^*} n p_n < +\infty$: the criterion in Theorem~\ref{theo:maindense} simply bowls down to the divergence of the harmonic series. We make the finite mean assumption to proceed.

The main virtue of this construction is that it provides a complete statistical model which we summarize: we take a renewal sequence $(D_k)_{\kN^*}$, set $F_k:=\psan (v_l)_{l\in \intt{1,D_k}}$, choose a random Matryoshka in $\underline{F}$, $(E_k)_{\kN^*}$ with $\dim E_k=k$ and take as Hamiltonian
\[
\hat{H}=\sum_{\lN^*} \omega_l |l\rangle \langle l|,
\] where $|l\rangle=u_l$, i.e., $\sum_{l\in \intt{1,k}}|l\rangle \langle l|$ is the orthogonal projector on $E_k$ and $\omega_k=\omega D_k$.

This model neglects all correlations between levels but for those which come from nearest level (the gap) correlation, as described by the common distribution of the $\dim F_k - \dim F_{k-1}$, $\kN ^*$. An interesting limit occurs when $\omega$ goes to $0$ but the gap itself has a limiting distribution.

We have insisted enough on the artificial features of this construction, so let us conclude on a more optimistic note. First the model at least exhibits correlations between the statistical properties of the eigenvalues and of the eigenstates. Second, as the distribution of the gap is our choice, one can in principle study the influence of level repulsion, expected in generic quantum system, and compare it with Poissonian statistics expected in integrable systems for instance. Third, the renewal structure ensures that the density of eigenvalues is asymptotically flat, so there is no need to straighten it; we could even take the first gap using the stationary measure, leading to a flat density everywhere.

We plan to return to the study of the statistical properties of this model in forthcoming works.

\appendix

\section{The measure on subspaces} \label{appsec:messubsp}

In this appendix we recall in an informal way a few salient features of $\mu_F^m$, the unitary invariant probability measure on the set of $m$-dimensional subspaces of $F$, a~finite-dimensional Hilbert space of dimension~$n$. We do not address uniqueness of $\mu_F^m$.

Our starting point is the Haar probability measure, $\mu_F$ on $U(F)$,\footnote{We apologize for a clash in notation with the main text, where $U$ denotes a \rv.} the group of unitary transformations of $F$. We fix an orthonormal basis $(v_1,\dots,v_n)$ of $F$.

For $m\in \intt{1,n}$ we consider the map $\Pi_m\colon U(F) \to \mathbb{G}(F,m)$ (recall this is the set of $m$-dimensional subspaces of $F$) defined by $\Pi_m(g):=\psan (gv_l)_{l\in \intt{1,m}}$. If $U(F)$ and $\mathbb{G}(F,m)$ are endowed with there usual smooth manifold structure and the associated Borel \salg , this map is smooth (hence measurable). So it induces a probability measure $\mu_F \circ \Pi_m^{-1}$ on $\mathbb{G}(F,m)$: the measure of a measurable subset $B$ is of $\mathbb{G}(F,m)$ is defined to be $\mu_F \big(\Pi_m^{-1}(B)\big)$. This measure on $\mathbb{G}(F,m)$ is unitary invariant because the Haar measure on $U(F)$ is, and we denote if by $\mu_F^m$. The group $U(F)$ contains the (linear extension of the) permutation group acting on the basis $(v_1,\dots,v_n)$. So even without knowing about uniqueness, it is clear by symmetry that using an arbitrary $m$-subset of the basis $(v_1,\dots,v_n)$ instead of $(v_1,\dots,v_m)$ to define $\Pi_m$ would lead to the same measure.

If $f\colon \mathbb{G}(F,m)\to \mathbb{R}$ is $\mu_F^m$-integrable, the definition of $\mu_F^m$ leads to
\[
\int_{E\in \mathbb{G}(F,m)} f(E) \, {\rm d}\mu_F^m(E):= \int_{g\in U(F)} f(\Pi_m(g)) \, {\rm d}\mu_F(g),
\]
and we apply this formula for special classes of functions $f$. Fix $v\in F$ and let $h_v(E)$ be the orthogonal projection of $v$ on the $m$-dimensional subspace $E$ of $F$. For $g\in U(F)$ we compute that
\[
\norm{h_v(\Pi_m(g))}^2 = |\aver{gv_1 ,v}|^2 +\cdots+ |\aver{gv_m ,v}|^2.
\]
Thus
\begin{align*}
\int_{E\in \mathbb{G}(F,m)} \norm{h_v(E)}^2 \, {\rm d}\mu_F^m(E)&:= \int_{g\in U(F)} \norm{h_v(\Pi_m(g))}^2 \, {\rm d}\mu_F(g)
\\
&
=\sum_{l\in \intt{1,m}} \int_{g\in U(F)} |\aver{gv_l ,v}|^2 \, {\rm d}\mu_F(g).
\end{align*}
By symmetry, it is plain that the $m$ summands on the right-hand side are equal, and for $m=n$ they sum to $\norm{v}^2$. Hence, for each $l\in \intt{1,n}$,
\[
\int_{g\in U(F)} |\aver{gv_l ,v }|^2 \, {\rm d}\mu_F(g) =\norm{v}^2\frac{1}{n}.
\]
Consequently
\[
\int_{E\in \mathbb{G}(F,m)} \norm{h_v(E)}^2 \, {\rm d}\mu_F^m(E)=\norm{v}^2\frac{m}{n}.
\]
Now $\int_{E\in \mathbb{G}(F,m)} f (E) \, {\rm d}\mu_F^m(E)$ is simply the expectation of (the random variable) $f$ under the probability measure $\mu_F^m$ and $h_v(E)$ is what we denoted by $U$ in the main text, so we retrieve the formula $\expec{\norm{U}^2}=\norm{v}^2\frac{m}{n}$.

\subsection*{Acknowledgments}

I thank Denis Bernard for discussions about the eigenstate thermalization hypothesis that finally led me to this study, and Philippe Biane for discussions on some technical points.

\pdfbookmark[1]{References}{ref}
\LastPageEnding


\begin{thebibliography}{99}
\footnotesize\itemsep=0pt

\bibitem{PhysRevA.43.2046}
Deutsch J.M., Quantum statistical mechanics in a closed system, \href{https://doi.org/10.1103/PhysRevA.43.2046}{\textit{Phys.
 Rev.~A}} \textbf{43} (1991), 2046--2049.

\bibitem{Ross}
Ross S., A first course in probability, 10th~ed., Pearson, 2019.

\bibitem{PhysRevE.50.888}
Srednicki M., Chaos and quantum thermalization, \href{https://doi.org/10.1103/PhysRevE.50.888}{\textit{Phys. Rev.~E}}
 \textbf{50} (1994), 888--901, \href{https://arxiv.org/abs/cond-mat/9403051}{arXiv:cond-mat/9403051}.

\end{thebibliography}
\end{document}